\documentclass{ejm}

\usepackage{hyperref}

\usepackage{mathrsfs}
\usepackage{tikz}
\usepackage{tkz-graph}
\usepackage{enumitem}
\setlistdepth{9}

\setlist[itemize,1]{label=$\,\bullet \,\,$}
\setlist[itemize,2]{label=$\,\bullet \,\,$}
\setlist[itemize,3]{label=$\,\bullet \,\,$}
\setlist[itemize,4]{label=$\,\bullet \,\,$}
\setlist[itemize,5]{label=$\,\bullet \,\,$}
\setlist[itemize,6]{label=$\,\bullet \,\,$}
\setlist[itemize,7]{label=$\,\bullet \,\,$}
\setlist[itemize,8]{label=$\,\bullet \,\,$}
\setlist[itemize,9]{label=$\,\bullet \,\,$}

\renewlist{itemize}{itemize}{9}

\let\realverbatim=\verbatim
\let\realendverbatim=\endverbatim
\renewcommand\verbatim{\par\addvspace{6pt plus 2pt minus 1pt}\realverbatim}
\renewcommand\endverbatim{\realendverbatim\addvspace{6pt plus 2pt minus 1pt}}

\ifprodtf \else
  \checkfont{eurm10}
  \iffontfound
    \IfFileExists{upmath.sty}
      {\typeout{^^JFound AMS Euler Roman fonts on the system,
                   using the 'upmath' package.^^J}%
       \usepackage{upmath}}
      {\typeout{^^JFound AMS Euler Roman fonts on the system, but you
                   don't seem to have the}%
       \typeout{'upmath' package installed. EJM.cls can take advantage
                 of these fonts,^^Jif you use 'upmath' package.^^J}%
       \providecommand\mathup[1]{##1}%
       \providecommand\mathbup[1]{##1}%
      }
  \else
    \providecommand\mathup[1]{#1}%
    \providecommand\mathbup[1]{#1}%
  \fi
\fi

\ifprodtf \else
  \checkfont{msam10}
  \iffontfound
    \IfFileExists{amssymb.sty}
      {\typeout{^^JFound AMS Symbol fonts on the system, using the
                'amssymb' package.^^J}%
       \usepackage{amssymb}%
         \let\leq=\leqslant
         \let\geq=\geqslant
      }{}
  \fi
\fi

\ifprodtf \else
  \IfFileExists{amsbsy.sty}
    {\typeout{^^JFound the 'amsbsy' package on the system, using it.^^J}%
     \usepackage{amsbsy}}
    {\providecommand\boldsymbol[1]{\mbox{\boldmath $##1$}}}
\fi

\newsavebox{\astrutbox}
\sbox{\astrutbox}{\rule[-5pt]{0pt}{20pt}}

\newcommand\etal{\mbox{\textit{et al.}}}

\newtheorem{thm}{Theorem}[section]
\newdefinition{defn}[thm]{Definition}
\newdefinition{example}[thm]{Example}
\newdefinition{lem}[thm]{Lemma}
\newdefinition{rem}[thm]{Remark}

\title[Physically Feasible Decomposition of Engino\textsuperscript{$^{\circledR}$}
Toy Models]{Physically Feasible Decomposition of Engino\textsuperscript{$^{\circledR}$}
Toy Models: A Graph Theoretic Approach}

\author[E.N. Antoniou et al.]{%
 E. N.  ANTONIOU$\,^1$,  A.  ARA\'UJO$\,^2$, 
 M. D. BUSTAMANTE$\,^3$ and
 A. GIBALI$\,^4$
}

\affiliation{%
  $^1\,${Department of Information Technology, Alexander Technological Educational Institute of Thessaloniki, Greece}\\
    email\textup{\nocorr: \texttt{antoniou@it.teithe.gr}}\\
  $^2\,${CMUC, Department of Mathematics, University of Coimbra, Portugal}\\
    email\textup{\nocorr: \texttt{alma@mat.uc.pt}}\\
  $^3\,${Institute for Discovery, School of Mathematics and Statistics, University College Dublin, Belfield D4, Ireland}\\
    email\textup{\nocorr: \texttt{miguel.bustamante@ucd.ie}}\\
  $^4\,${Department of Mathematics, ORT Braude College, Israel}\\
    email\textup{\nocorr: \texttt{avivg@braude.ac.il}}}

\date{27 July 2017}

\begin{document}

\label{firstpage}
\maketitle

\begin{abstract}
During the 125th European Study Group with Industry held in Limassol,
Cyprus, 5-9 December 2016, one of the participating companies, Engino.net 
Ltd, posed a very interesting challenge to the members of the study
group. Engino.net Ltd is a Cypriot company, founded in
2004, that produces a series of toy sets -- the Engino$^{\circledR}$ toy sets --  consisting of a number
of building blocks which can be assembled by pupils to compose toy
models. Depending on the contents of a particular toy set, the company
has developed a number of models that can be built utilizing the blocks
present in the set, however the production of a step-by-step assembly
manual for each model could only be done manually. The goal of the
challenge posed by the company was to implement a procedure to automatically
generate the assembly instructions for a given toy. In the present
paper we propose a graph-theoretic approach to model the problem and
provide a series of results to solve it by employing modified versions
of well established algorithms in graph theory. An algorithmic procedure
to obtain a hierarchical, physically feasible decomposition of a given toy
model, from which a series of step-by-step assembly instructions can be 
recovered, is proposed. 
\end{abstract}

\begin{keywords}
62P30, 05C90, 68R10, 94C15.
\end{keywords}

\section{Introduction}\label{sec:intro}

Engino$^{\circledR}$ toy models are created by assembling small blocks
or bricks together, with the purpose of helping pupils build technological
models creatively and easily so that they can experiment and learn
about science and technology in a playful way. Each of the toy sets
produced by Engino.net Ltd has a specific number of blocks
that can be assembled into many different models. It has been observed
that the creative potential of each toy set increases exponentially
as the number of blocks in the set increase. This is due to the patented
design of the Engino$^{\circledR}$ blocks that allow connectivity
on many directions in three-dimensional space simultaneously.

To demonstrate the creative potential of its toy sets, the company
has developed a large number of toy models that can be built using the
contents of the set. The ingredients and the connections required
to obtain each particular toy model has been recorded in a database
system. Despite the detailed recording, the production of step-by-step
instructions for the assembly of a particular toy model has been
proved to be a tedious task that has to be accomplished manually.
This is mainly due to the three-dimensional nature of the models and
the complexity of the interconnections between the blocks, which in
many cases impose a particular order in the steps that have to be
taken to assembly the structure. The goal of the challenge posed by
the company during the 125th European Study Group with Industry was
the development of an automatic procedure able to produce step-by-step
assembly instructions manual for every toy model that has been recorded
in the company's databases. 

To accomplish this task, we propose a graph-theoretic approach. Given
a toy model, we associate with it a directed graph whose vertices
correspond to the building blocks of the model and whose edges represent
physical connections between two blocks (see \cite{Peysakhov} and references 
therein). Moreover, in order to partially capture the actual geometry of the 
toy model, every edge of the graph is labeled with a vector showing the 
direction of the underlying physical connection in 3D space. This labeling 
of the edges provides an adequate description of the geometry of the model, 
for the purposes of our application. It should be however noted, that the exact
geometry of the (possibly multiple) connections between individual blocks of
a particular model has been recorded in full detail and this information is
available at the final stage of the assembly instructions generation.

With this setup, in order to produce the assembly instructions of
a given model, we first apply the reverse process recursively.
Given a description of a toy model, and hence its associated graph,
we develop a method to break it apart into clusters of blocks in a
manner that is physically possible. In what follows we call this procedure
a \emph{Physically Feasible Decomposition (PFD)} of the model. The
result of such a decomposition is a collection of sub-models or components
on which the method can be recursively applied until no further decompositions
are possible. Thus, a characterization of PFD of a model is of fundamental importance in the decomposition procedure. The outcome of this separation process is a hierarchical
tree structure of components, whose nodes can traversed in a postorder fashion, to generate an ordered sequence of nodes, which in turn dictate a series of step-by-step assembly instructions. 

The problem of determining a series of steps required to decompose a complex structure into its constituent components has been the subject of several studies dating back to the 1980s. This class of problems is termed \emph{disassembly sequencing} and depending on the nature of the underlying structure, a number of different approaches have been employed (see \cite{Lambert2003} for an extensive survey). The motivation behind the study of disassembly sequencing originates mainly from the fact that by reversing the steps of a disassembly sequence, one can obtain an assembly procedure of the structure under study. In this respect, disassembly sequences are closely related to the automated generation of assembly instructions of complex structures (see for instance \cite{Agrawala2003,Li2008,Hsu2011}). 

The procedure proposed in the present paper can be compared to the one presented in \cite{Li2008} for the computation of a hierarchical explosion graph. Contrary to the approach used in \cite{Agrawala2003,Li2008} for the construction of the explosion graph, which detaches individual parts one-by-one from the structure, and in turn applies a search strategy for the extraction of the hierarchy of components, our method produces directly a physically feasible decomposition into components along a given spatial direction. As shown in Section \ref{sec:MaximalPFD}, a maximal physically feasible decomposition can be obtained using well known linear-time algorithms and the recursive application of this procedure results in a hierarchical decomposition which is comparable to the hierarchical explosion graph in \cite{Li2008}. 

The contents of the paper are organized as follows: In Section \ref{sec:Prerequisites}, we briefly
recall some basic concepts and facts from graph theory required for
the development of our results in the sequel. In the subsequent section,
we present the proposed graph theoretic framework and through a series
of motivating examples we introduce the notion of a \emph{Physically
Feasible Decomposition (PFD)} of a toy model. In the same section,
we also define the \emph{Component Connectivity Graph (CCG)} implied
by the removal of a set of edges of the model's graph and show that
such a removal gives rise to a PFD if and only if the corresponding
CCG is a directed acyclic graph. In Section \ref{sec:MaximalPFD}, we define maximal PFDs
along a given direction and show that such decompositions can be obtained
by applying well established, linear-time, algorithms used
for the discovery of strongly connected components in directed graphs.
In Section \ref{sec:HierarchicalPFD} we outline an algorithmic procedure 
to obtain a hierarchical decomposition of a given toy model, using as intermediate steps
for such a decomposition, maximal PFDs along appropriately chosen spatial directions. 
Moreover, at the end of section \ref{sec:HierarchicalPFD}, the resulting hierarchical 
decomposition of the model is utilized to recover a series of assembly instructions.
Finally, in Section \ref{sec:Conclusions} we review and summarize our results. 

\section{\label{sec:Prerequisites}Graph Theory Prerequisites}

In this section, we review a number of definitions and facts from graph
theory that will be instrumental in the sequel. Most of these definitions
and results can be found in \cite{BangGutin2008,BondyAndMurty1976}.

A \emph{directed graph} $G$, denoted by $G(V,E)$, is an ordered
pair of sets $(V,E)$ where:
\begin{itemize}
\item $V$ is the set \emph{vertices} or nodes of $G$;
\item $E$ is the set of \emph{directed edges} consisting of directed pairs
$(u,v)$, where $u, v \in V$.
\end{itemize}
Moreover, if $E$ is allowed to be a multiset instead of a set, then
$G(V,E)$ is a \emph{directed multigraph}. On the other hand, if pairs of the form $(v,v)$, (called \emph{loops}) are
not allowed in $E$, then $G(V,E)$ is a \emph{directed simple graph}. Similar
definitions can be given in case the edge set (multiset), has as elements
undirected pairs of vertices. In such a case the (multi)graph is called
\emph{undirected}.

A graph $G_{1}(V_{1},E_{1})$ is a \emph{subgraph} of a given graph $G(V,E)$
if $V_{1}\subseteq V$ and $E_{1}\subseteq E$ consists exclusively
of edges having both its endpoints in $V_{1}$. Moreover, for $V_{1}\subseteq V$,
we define the\emph{ induced subgraph} $G[V_{1}]$ as the subgraph
of $G(V,E)$, whose vertex set is $V_{1}$ and its edge set is the
set of all edges of $E$, having both their endpoints in $V_{1}$.

In a directed graph $G(V,E)$, a \emph{directed (resp. undirected)
path} of length $k$, starting from $v_{0}$ and ending to $v_{k}$,
is a sequence of vertices $v_{0},v_{1},\ldots,v_{k}$, such that$(v_{i},v_{i+1})\in E$
($resp.(v_{i},v_{i+1})\in E$ or $(v_{i+1},v_{i})\in E$), for all
$0\leq i<k$. In case $v_{0}=v_{k}$ and $k>0$ the path is called
a \emph{directed (resp. undirected) cycle}. A vertex $t\in V$ is
said to be \emph{reachable} from $s\in V$, if there exists a directed
path from $s$ and to $t$. 

A directed graph $G(V,E)$ is set to be a \emph{Directed Acyclic Graph (DAG)}, if it contains no directed cycles, or equivalently, if there
exists no vertex in $V$ which is non-trivially reachable from itself.
A \emph{topological ordering} of the vertices of a directed graph
$G(V,E)$ is a total ordering of its vertices $v_{1},v_{2},\ldots,v_{n}$,
such that for all $(v_{i},v_{j})\in E$, $i\leq j$ holds. 
\begin{thm}
A directed graph $G(V,E)$ is acyclic if and only if a topological
ordering of its vertices exists.
\end{thm}
A directed graph $G(V,E)$ is called\emph{ strongly (resp. weakly)
connected} if for every pair of vertices $u\in V$, $v\in V$, there
exists a directed (resp. undirected) path from $u$ to $v$. A maximal
strongly (resp. weakly) connected subgraph of a graph, i.e. a strongly
connected subgraph which is not a proper subgraph of any other strongly
connected subgraph, is called a \emph{strongly (resp. weakly) connected
component}. 

The \emph{condensation} of a directed graph $G(V,E)$ is a directed
graph $G_{co}(V_{co},E_{co})$, with:
\begin{itemize}
\item $V_{co}=\{C_{i}:C_{i}\textrm{ is a strongly connected component of }G(V,E)\}$; 
\item $E_{co}=\{(C_{i},C_{j}):\exists(u,v)\in E\textrm{ such that }u\in C_{i},v\in C_{j}\}$.
\end{itemize}
\begin{thm}
The condensation of any directed graph $G(V,E)$ is a directed acyclic
graph.
\end{thm}

A \emph{tree} is an undirected graph in which every pair of vertices
is connected via a unique path. A \emph{rooted tree} is a tree having
one particular vertex designated as its \emph{root node}. An \emph{ordered
tree} is a rooted tree in which an ordering is specified for the children
of each vertex. A \emph{binary tree} is a rooted tree in which every
vertex has at most two children. A binary tree is \emph{full} if every
node has either zero or two children. 

\section{\label{sec:PFD}Physically Feasible Decomposition of Toy Models}

We now present the proposed framework for the solution of the decomposition
problem discussed above based on a graph-theoretic approach. Given
a toy model $\mathscr{M}$, we associate to it a directed graph $G(V,E)$,
where:
\begin{itemize}
\item $V=\{v_{1},v_{2},\ldots,v_{n}\}$ is the vertex set of $G$ with each
vertex $v_{i}$ corresponding to a block of $\mathscr{M}$;
\item $E=\{(u,v): u, v\in V\}$
is the edge set of $G$ with each directed edge representing a connection
between two blocks of the model. 
\end{itemize}
Every physical connection between two blocks of the model can be aligned
in space to one particular direction vector, chosen out of a finite
collection of directions. For instance, if a model uses only perpendicular
connections between its blocks in 3D space, we can identify three
direction vectors $\hat{i},\hat{j},\hat{k}$ along which all connections
can be aligned. A connection between two blocks of the model $u,v$,
aligned to a particular direction $\hat{d}$ in physical space, gives
rise to a directed edge $(u,v)\in E$, if the vector from $u$ to
$v$ points towards the same direction as $\hat{d}$. 

Assuming that all the connections of the model $\mathscr{M}$ correspond
to $p$, not necessarily orthogonal, distinct spatial directions $\hat{d_{i}}$, we can partition
the edge set $E$ into a family of $p$ mutually disjoint sets $E_{i}$,
$i=1,2,\ldots,p$, each of which contains the edges associated to
connections sharing the same direction in space. It should be noted that the physical connections between the blocks of the toy model $\mathscr{M}$, are assumed to be fixed, meaning that the resulting construction is rigid and contains no moving or rotating parts. Thus, the only possible way to separate two connected blocks is to apply opposite forces along the physical direction $\hat{d_{i}}$ associated to the connection, provided that resulting the displacement is physically feasible in the sense described in the paragraph that follows. At first this may seem to be a rather restrictive assumption with respect to the types of toy models it allows to be constructed, as there are many actual toy compositions in Engino's collection involving moving or rotating parts. However, as discussed with representatives of the company during the 125th ESGI meeting, in most such cases the moving or rotating parts can either be considered as separate rigid submodels (e.g. a two wheel and axle submodel), or their connection to the rest of the model is non-fixed (e.g. a pinned joint), allowing them to be detached from it by pulling them along some non-blocking direction. In view of this setup, we propose the following principle to describe the conditions under which a disconnection of two blocks is physically possible.

\begin{description}
\item [{Physically Feasible Disconnection of two blocks:} ] In order to
disconnect two blocks corresponding to vertices $v_{1},v_{2}\in V$,
connected via an edge $(v_{1},v_{2})\in E$ aligned to a given spatial
direction $\hat{d_{i}}$, the blocks $v_{1},v_{2}$ must be able to
be displaced along the directions $-\hat{d_{i}},\hat{d}_{i}$ respectively,
when appropriate opposite forces are applied on the blocks. 
\end{description}

The idea behind the above principle is illustrated in the following
example. 
\begin{example}
Consider the blocks shown in Figure \ref{fig:TwoPFDBlocks}. 
\begin{figure}[h]
\begin{centering}
\includegraphics[scale=0.7]{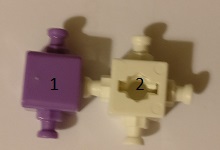}\  \includegraphics[scale=0.7]{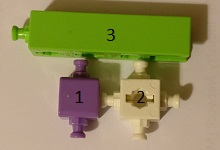} 
\par\end{centering}
\caption{Two blocks that can be disconnected (left); blocks $1,2$ cannot be disconnected (right).}
\label{fig:TwoPFDBlocks} 
\end{figure}
In the left side, the blocks $1,2$ can be disconnected using two opposite horizontal forces, since their application on the two blocks will result in displacements
along the horizontal direction. If a third block is added as shown in right side of Figure 
\ref{fig:TwoPFDBlocks}, then the blocks $1,2$ cannot be disconnected by applying
on them opposite horizontal forces, since their displacement is blocked
by their vertical connections to the block number 3.
\end{example}

The idea of disconnecting two blocks of the model in a physically
feasible manner can be easily generalized to describe the corresponding
decomposition of a model into two submodels. In general, the removal
of a set of edges along a given direction may result into a decomposition
of the graph of the model into two or more weakly connected components.
However, not all such removals can be actually applied on the physical
model to decompose it into two or more submodels. This is due to the
fact that in certain cases the physical displacement of the resulting
weakly connected components of the model is blocked by other physical
connections, due to the presence of edges not removed in the current
phase.

We can extend the principle of Physically Feasible Disconnection,
introduced above, to the case of the separation of two weakly connected
components.

\begin{description}
\item [{Physically Feasible Decomposition into two components:}] The removal
of a set of edges, aligned to a particular space direction $\hat{d_{i}}$,
is physically feasible, if and only if the two resulting weakly connected
components are able to be displaced along the directions $-\hat{d_{i}},\hat{d}_{i}$
respectively, when appropriate opposite forces are applied on these
blocks. 
\end{description}

For brevity, in what follows, we shall call this decomposition a
\emph{2\textendash PFD of the model}. The above decomposition is equivalent
to assuming that, during the separation process, each of the two weakly
connected components behaves like a single block, but unlike the single
blocks case, it is possible to have multiple parallel connections
between them. 

Our next goal is to obtain a characterization of 2\textendash PFD's 
that are possible along a given direction. In this respect it is instrumental
to introduce the notion of the \emph{Component Connectivity Graph}
of a model $\mathscr{M}$, implied by the removal of a set of co-linear
edges, which provides a higher level view of the decomposition.
\begin{defn}[Component Connectivity Graph]
 Let $G(V,E)$ be the graph associated to a model $\mathscr{M}$
and $\bar{E_{i}}\subseteq E_{i}$ a non-empty set of edges, where
$E_{i}$ is the set of all edges of $G(V,E)$ along the spatial direction
$\hat{d}_{i}$. The \emph{Component Connectivity Graph (CCG)}, implied
by the removal of the edges of $\bar{E_{i}}$, is a directed graph,
$G_{C}(V_{C},E_{C})$, whose vertices are the weakly connected components
$C_{i}$, $i=1,2,\ldots,k$, into which $G(V,E)$ is partitioned with
the removal of the edges of $\bar{E_{i}}$. Two components $C_{i},C_{j}$
are connected via an edge $(C_{i},C_{j})\in E_{C}$ if and only if
$i\neq j$ and there exists an edge $(v,u)\in\bar{E_{i}}$ with $v\in C_{i}$
and $u\in C_{j}$.
\end{defn}

We should note that according to the above definition the CCG implied
by the removal of a set of edges $\bar{E_{i}}\subseteq E_{i}$ is
a simple directed graph, since by construction it cannot contain neither
loops nor parallel edges sharing the same source and target vertices.
The above ideas are illustrated in the following example.

\begin{example}
\label{exa:2PFDvsNoPFD}Consider the model shown in Figure \ref{fig:Model1_Photo}. The graph $G(V,E)$ of the model is depicted in Figure \ref{fig:Model1_Graph}, where $\hat{d}_{1},\hat{d}_{2}$ are respectively the horizontal (left
- right) and vertical (bottom - up) direction vectors.
\begin{figure}[h]
\begin{centering}
\includegraphics[scale=0.3]{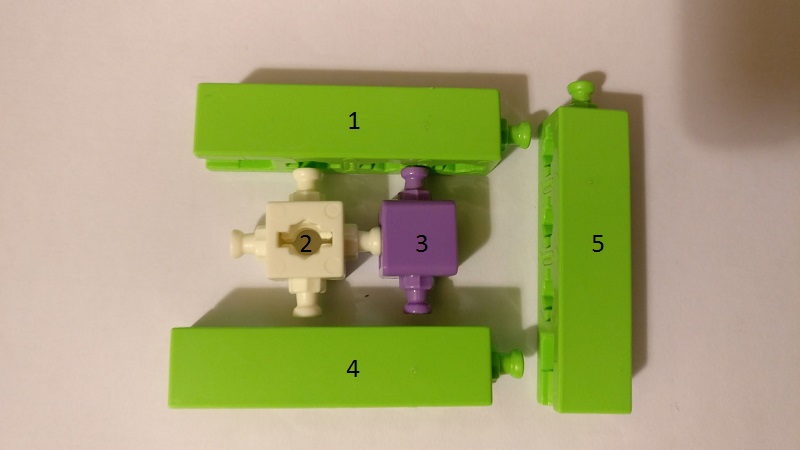} 
\par\end{centering}
\caption{A picture of the actual model.}
\label{fig:Model1_Photo} 
\end{figure}
\begin{figure}[h]
\begin{centering}
\begin{center}
\begin{tikzpicture}
\SetGraphUnit{2}
  \GraphInit[vstyle=Normal]

  \Vertex[L=$1$]{V1}
  \SOWE[L=$2$](V1){V2}
  \SOEA[L=$4$](V2){V4}
  \SOEA[L=$3$](V1){V3}
  \EA[L=$5$](V3){V5}

  \Edge[label=$\hat{d}_1$,style=->](V2)(V3)
  \Edge[label=$\hat{d}_1$,style=->](V1)(V5)
  \Edge[label=$\hat{d}_1$,style=->](V4)(V5)
  \Edge[label=$\hat{d}_2$,style=->](V2)(V1)
  \Edge[label=$\hat{d}_2$,style=->](V4)(V2)
  \Edge[label=$\hat{d}_2$,style=->](V3)(V1)
  \Edge[label=$\hat{d}_2$,style=->](V4)(V3)
\end{tikzpicture}
\par\end{center}
\par\end{centering}
\caption{The graph $G(V,E)$ of the model.}
\label{fig:Model1_Graph} 
\end{figure}
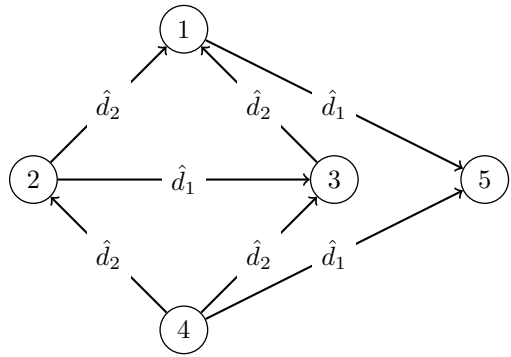

If we remove all edges along the horizontal direction, i.e. edges
$(2,3)$, $(1,5)$ and $(4,5)$, the graph is decomposed into two
weakly connected components $C_{1}=\{1,2,3,4\}$ and $C_{2}=\{5\}$
as shown in Figure \ref{fig:Model1_Graph_HorRem}, and the implied CCG by this removal of edges is shown in Figure \ref{fig:Model1_CCG_2PFD}.
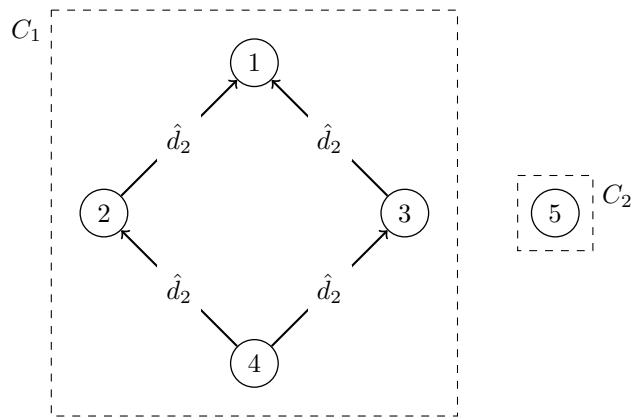
\begin{figure}[h]
\begin{centering}
\begin{center}
\begin{tikzpicture}
\SetGraphUnit{2}
  \GraphInit[vstyle=Normal]

\draw[dashed] (-2.7,0.7) -- (2.7,0.7) -- (2.7,-4.7) 
		-- (-2.7,-4.7) -- (-2.7,0.7) node[anchor=north east] {$C_1$};

\draw[dashed] (3.5,-1.5) -- (4.5,-1.5) node[anchor=north west] {$C_2$} -- (4.5,-2.5) 
		-- (3.5,-2.5) -- (3.5,-1.5);

  \Vertex[L=$1$]{V1}
  \SOWE[L=$2$](V1){V2}
  \SOEA[L=$4$](V2){V4}
  \SOEA[L=$3$](V1){V3}
  \EA[L=$5$](V3){V5}

  \Edge[label=$\hat{d}_2$,style=->](V2)(V1)
  \Edge[label=$\hat{d}_2$,style=->](V4)(V2)
  \Edge[label=$\hat{d}_2$,style=->](V3)(V1)
  \Edge[label=$\hat{d}_2$,style=->](V4)(V3)
\end{tikzpicture}
\par\end{center}
\par\end{centering}
\caption{The graph $G(V,E)$ after the removal of all edges along $\hat{d}_{1}$.}
\label{fig:Model1_Graph_HorRem} 
\end{figure}
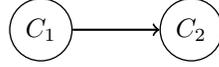
\begin{figure}[h]
\begin{centering}
\begin{center}
\begin{tikzpicture}
\SetGraphUnit{2}
  \GraphInit[vstyle=Normal]

  \Vertex[L=$C_1$]{C1}
  \EA[L=$C_2$](C1){C2}

  \Edge[style=->](C1)(C2)

\end{tikzpicture}
\par\end{center}
\par\end{centering}
\caption{The CCG of the model after the removal of all edges along $\hat{d}_{1}$.}
\label{fig:Model1_CCG_2PFD} 
\end{figure}
Clearly, nothing prevents the displacement of the two components $C_{1},C_{2}$
from moving towards $-\hat{d}_{1},\hat{d}_{1}$ respectively, when
appropriate horizontal forces are applied on them. Thus, the removal
of all horizontal edges implies a 2\textendash PFD of the model. 

On the other hand, if we choose to remove all edges along $\hat{d}_{2}$,
we end up with the weakly connected components $C_{1}^{\prime},C_{2}^{\prime}$
shown in Figure \ref{fig:Model1_Graph_VerRem}, 
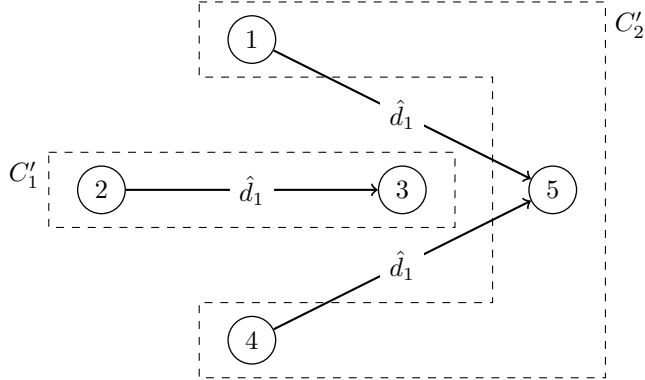
\begin{figure}[h]
\begin{centering}
\begin{center}
\begin{tikzpicture}
\SetGraphUnit{2}
  \GraphInit[vstyle=Normal]

\draw[dashed] (-2.7,-1.5) -- (2.7,-1.5) -- (2.7,-2.5) 
		-- (-2.7,-2.5) -- (-2.7,-1.5) 
		node[anchor=north east] {$C^\prime_1$};

\draw[dashed] (-0.7,0.5) -- (4.7,0.5) 
		node[anchor=north west] {$C^\prime_2$}
		-- (4.7,-4.5) -- (-0.7,-4.5)
		-- (-0.7,-3.5) -- (3.2,-3.5)
		-- (3.2,-0.5) -- (-0.7,-0.5) -- (-0.7,0.5);

  \Vertex[L=$1$]{V1}
  \SOWE[L=$2$](V1){V2}
  \SOEA[L=$4$](V2){V4}
  \SOEA[L=$3$](V1){V3}
  \EA[L=$5$](V3){V5}

  \Edge[label=$\hat{d}_1$,style=->](V2)(V3)
  \Edge[label=$\hat{d}_1$,style=->](V1)(V5)
  \Edge[label=$\hat{d}_1$,style=->](V4)(V5)
\end{tikzpicture}
\par\end{center}
\par\end{centering}
\caption{The graph $G(V,E)$ after the removal of all edges along $\hat{d}_{2}$.}
\label{fig:Model1_Graph_VerRem} 
\end{figure}
and the corresponding CCG is the one in Figure \ref{fig:Model1_CCG_NoPFD}.
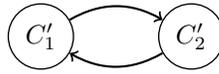
\begin{figure}[h]
\begin{centering}
\begin{center}
\begin{tikzpicture}
\SetGraphUnit{2}
  \GraphInit[vstyle=Normal]

  \Vertex[L=$C^\prime_1$]{C1}
  \EA[L=$C^\prime_2$](C1){C2}

  \Edge[style={bend left, ->}](C1)(C2)
  \Edge[style={bend left, ->}](C2)(C1)

\end{tikzpicture}
\par\end{center}
\par\end{centering}
\caption{The CCG of the model after the removal of all edges along $\hat{d}_{2}$.}
\label{fig:Model1_CCG_NoPFD} 
\end{figure}
Despite the fact that the removal of the four vertical edges separates
the graph into two weakly connected components, it is clear that such
a decomposition is not physically feasible. Obviously, the blocks
$2,3$ of $C_{1}^{\prime}$ cannot be displaced vertically, because
they are ``trapped'' between the components $1,4$ of $C_{2}^{\prime}$.
\end{example}

In view of the decomposition along the spatial direction $\hat{d}_1$ shown in Example \ref{exa:2PFDvsNoPFD}, it becomes apparent that not all the edges removed correspond to a physically feasible disconnection of two blocks. This is the case with the edge $(2,3)$ in the graph of Example \ref{exa:2PFDvsNoPFD}, which does not appear in Figure  \ref{fig:Model1_Graph_HorRem} due to its removal. Despite the fact that this edge can be theoretically removed during the removal of all edges along $\hat{d}_1$, the blocks $2,3$ cannot be disconnected because the perpendicular connections with blocks $1,4$ obstruct their horizontal displacement. On the other hand, the edges $(1,5)$, $(4,5)$ obviously contribute actively on the decomposition of the graph into two components $C_{1}$ and $C_{2}$, shown in Figure \ref{fig:Model1_Graph_HorRem_PhysRemovableOnly}. The distinguishing property between these two types of edges is that the former has both its endpoints on the same weakly connected component after the removal of all edges along $\hat{d}_1$, while each of the latter type of edges have their start and end points lying on distinct components. The edges that actively contribute to the formation of weakly connected components of a given CCG, will be called \emph{physically removable} for the given CCG. A maximal subset of physically removable edges, along a given spatial direction, can be successfully computed using the technique presented in Section \ref{sec:MaximalPFD}.

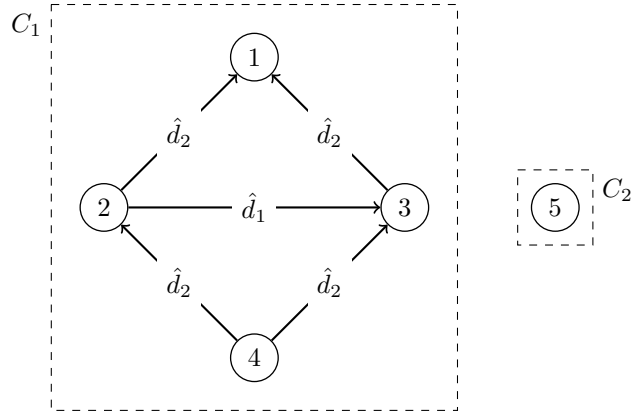
\begin{figure}[h]
	\begin{centering}
		\begin{center}
			\begin{tikzpicture}
			\SetGraphUnit{2}
			\GraphInit[vstyle=Normal]
			
			\draw[dashed] (-2.7,0.7) -- (2.7,0.7) -- (2.7,-4.7) 
			-- (-2.7,-4.7) -- (-2.7,0.7) node[anchor=north east] {$C_1$};
			
			\draw[dashed] (3.5,-1.5) -- (4.5,-1.5) node[anchor=north west] {$C_2$} -- (4.5,-2.5) 
			-- (3.5,-2.5) -- (3.5,-1.5);
			
			\Vertex[L=$1$]{V1}
			\SOWE[L=$2$](V1){V2}
			\SOEA[L=$4$](V2){V4}
			\SOEA[L=$3$](V1){V3}
			\EA[L=$5$](V3){V5}

			\Edge[label=$\hat{d}_1$,style=->](V2)(V3)
									
			\Edge[label=$\hat{d}_2$,style=->](V2)(V1)
			\Edge[label=$\hat{d}_2$,style=->](V4)(V2)
			\Edge[label=$\hat{d}_2$,style=->](V3)(V1)
			\Edge[label=$\hat{d}_2$,style=->](V4)(V3)
			\end{tikzpicture}
			\par\end{center}
		\par\end{centering}
	\caption{The graph $G(V,E)$ after the removal of all edges along $\hat{d}_{1}$ which are physically removable.}
	\label{fig:Model1_Graph_HorRem_PhysRemovableOnly} 
\end{figure}

Proceeding a step further, we can provide a characterization of 2\textendash PFD's in terms of
a particular property of the edges connecting the weakly connected components in the corresponding CCG.

\begin{lem}
\label{lem:2-PFDCondition}Let $\mathscr{M}$ be a toy model with
the associated directed graph $G(V,E)$. Assume that the removal
of a non-empty set of edges $\bar{E_{i}}\subseteq E_{i}$, where $E_{i}$
is the set of all edges of $G(V,E)$ along the direction $\hat{d}_{i}$,
gives rise to the CCG, $G_{C}(V_{C},E_{C})$, where $V_{C}=\{C_{1},C_{2}\}$.
Then, the removal of the edges of $\bar{E_{i}}$ is a \emph{2\textendash PFD
of $\mathscr{M}$ if and only if $E_{C}$ contains exactly one of
the edges $(C_{1},C_{2})$, $(C_{2},C_{1})$.}
\end{lem}
\begin{proof}
We first note that since $\bar{E_{i}}$ is non-empty, so is $E_{C}$.
Moreover, recall that $G_{C}(V_{C},E_{C})$ is simple, so, $E_{C}$
will either contain exactly one or both \emph{$(C_{1},C_{2})$, $(C_{2},C_{1})$.
}Assume now that $E_{C}$ contains both $(C_{1},C_{2})$ and $(C_{2},C_{1})$.
Then, due to the presence of $(C_{1},C_{2})$, in order to separate
$C_{1}$ from $C_{2}$ we should be able to displace $C_{1}$ towards
$-\hat{d}_{i}$ and $C_{2}$ towards $\hat{d}_{i}$, by applying appropriate
opposite forces on $C_{1}$ and $C_{2}$. On the other hand, due to
the presence of $(C_{2},C_{1})$, in order to accomplish the same
task, $C_{1}$ should be able to move towards $\hat{d}_{i}$ and $C_{2}$
towards $-\hat{d}_{i}$, using again appropriate opposite forces.
Obviously, neither $C_{1}$ nor $C_{2}$ can move simultaneously on
both spatial directions $-\hat{d}_{i}$, $\hat{d}_{i}$. Thus, the
removal of the edges of $\bar{E_{i}}$, is not a 2\textendash PFD,
which proves the ``only if'' part of the lemma.

Conversely, assume without loss of generality that $E_{C}$ contains
only $(C_{1},C_{2})$. This means that, in physical space, the components
$C_{1}$,$C_{2}$ are connected only on one side, leaving their
externally exposed sides free (see Figure \ref{fig:2-PFD}). Thus, removing
the edges of $\bar{E_{i}}$ connecting the vertices of $C_{1}$ to
those of $C_{2}$, will result in a 2\textendash PFD of the model,
since $C_{1}$ can be displaced towards the direction of $-\hat{d}_{i}$
and $C_{2}$ towards that of $\hat{d}_{i}$.
\begin{figure}[h]
\begin{centering}
\begin{center}
\begin{tikzpicture}
\node at (0.5,1) [] {$C_1$};
\draw (0, 0) -- (0, 2) -- (1, 2) --
        (1, 1.75) -- (0.5, 1.75) -- (0.5, 1.5) --
        (2, 1.5) -- (2, 1) -- (1, 1) --
        (1, 0.5) -- (1.5, 0.5) -- (1.5, 0) -- (0, 0) ;

\node at (4.5,1) [] {$C_2$};
\draw(3.5, 0) -- (5, 0) -- (5, 2) -- (3, 2) --
        (3, 1.75) -- (2.5, 1.75) -- (2.5, 1.5) --
        (4, 1.5) -- (4, 1) -- (3, 1) --
        (3, 0.5) -- (3.5, 0.5) -- (3.5, 0);

\draw[->, dashed] (1.25, 1.75) -- (2.4, 1.75);
\draw[->, dashed] (1.5, 0.75) -- (2.9, 0.75);
\draw[->, dashed] (1.75, 0.25) -- (3.15, 0.25);

\draw[->] (-0.2, 1) -- node[above] {$-\hat{d}_i$}  (-1, 1);
\draw[->] (5.2, 1) -- node[above] {$\hat{d}_i$}  (6, 1);
\end{tikzpicture}
\par\end{center}
\par\end{centering}
\caption{2\textendash PFD of $C_{1}$, $C_{2}$.}
\label{fig:2-PFD} 
\end{figure}
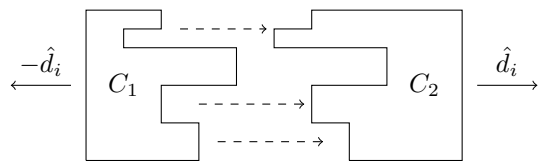
\end{proof}

Proceeding a step further we can generalize the idea of a Physically
Feasible Decomposition into the case where the removal of a set of
edges, along a particular spatial direction $\hat{d_{i}}$, separates
the model into more than two weakly connected components. Assume that
after the removal of a set of edges $\bar{E_{i}}\subseteq E_{i}$,
we end up with $k>2$ components. Such a decomposition is physically
feasible if we can obtain it by applying a 2\textendash PFD of the
original model by removing an appropriate subset of edges of $\bar{E_{i}}$,
and in turn by repeating 2\textendash PFD procedures on the resulting
submodels, recursively. A PFD giving rise to $k>2$ components, that
can be accomplished recursively by applying a series 2\textendash PFD's,
will be called a $k$\textendash PFD. 

The above idea is formalized in the following definition.
\begin{defn}[$k$\textendash PFD]
\label{def:k-PFD}Let $\mathscr{M}$ be a toy model and $G(V,E)$
its associated directed graph. Assume that the removal of
a non-empty set of edges $\bar{E_{i}}\subseteq E_{i}$, where $E_{i}$
is the set of all edges of $G(V,E)$ along the direction $\hat{d}_{i}$,
gives rise to the CCG, $G_{C}(V_{C},E_{C})$, consisting of $k\geq2$
weakly connected components. We say that the removal of the edges
$\bar{E_{i}}$ implies a $k$\textendash PFD of the model $\mathscr{M}$,
if there exists a set of edges $\bar{E}_{i}^{0}\subseteq\bar{E_{i}}$,
whose removal implies a 2\textendash PFD of $\mathscr{M}$ into $C_{1},C_{2}$,
for which exactly one of the following is true:
\begin{itemize}
\item $C_{1}\in V_{C}$ and $C_{2}\in V_{C}$;
\item $C_{1}\in V_{C}$ and the removal of all edges of $\bar{E_{i}}\backslash\bar{E}_{i}^{0}$
from $C_{2}$, implies its $(k-1)$\textendash PFD;
\item $C_{2}\in V_{C}$ and the removal of all edges of $\bar{E_{i}}\backslash\bar{E}_{i}^{0}$
from $C_{1}$, implies its $(k-1)$\textendash PFD;
\item $C_{j}\notin V_{C}$, for $j=1,2$ and appropriate removal of edges
of $\bar{E_{i}}\backslash\bar{E}_{i}^{0}$ from each one of them,
implies a $k_{1}$\textendash PFD of $C_{1}$ and a $k_{2}$\textendash PFD
of $C_{2}$, such that $k_{1}+k_{2}=k$. 
\end{itemize}
If the removal of any set of edges $\bar{E_{i}}\subseteq E_{i}$,
results in a CCG with only one weakly connected component, we say that we have
a \emph{ 1\textendash PFD} or a \emph{non PFD} of the model. 
\end{defn}
\begin{rem}
\label{rmk:k-PFD_as_Tree}The structure of a $k$\textendash PFD of
a model $\mathscr{M}$ can be represented by a full, ordered, binary
tree $T$, having as its root node the entire vertex set $V_{C}$.
The internal nodes of $T$ are subsets of $V_{C}$ corresponding
to weakly connected components of $G_{C}$ resulting in each step
of the recursive application of 2\textendash PFD's. Finally, the leaves
of $T$ are the singletons of $V_{C}$, that is, the components of
the CCG corresponding to the $k$\textendash PFD. Clearly, by construction
each node of $T$, will have either 0 or 2 children, thus $T$ is
full. Moreover, $T$ can be assumed to be ordered, that is, we distinguish
the left and the right child of each node. According to Lemma \ref{lem:2-PFDCondition},
every 2\textendash PFD separates a weakly connected component into
two child components, connected only in a single direction. In view
of this property we assign to the left child of each node in $T$,
the child component from which the edges originate, and to the right
child of the node in $T$, the component to which the edges terminate.
\end{rem}
Our aim is to identify those subsets of edges $\bar{E_{i}}\subseteq E_{i}$,
that is, sets of edges aligned to a spatial direction $\hat{d}_{i}$,
whose removal gives rise to a $k$\textendash PFD of the model. The
following theorem serves as a characterization of this property. 
\begin{thm}
\label{th:DAGImpliesPFD}Let $\mathscr{M}$ be a toy model and its
associated directed graph $G(V,E)$. Let further $G_{C}(V_{C},E_{C})$
be the CCG resulting after the removal of a non-empty set of edges
$\bar{E_{i}}\subseteq E_{i}$, where $E_{i}$ is the set of all edges
of $G(V,E)$ along the direction $\hat{d}_{i}$. The removal of the
edges $\bar{E_{i}}$ implies a $k$\textendash PFD, $k\geq2$ of the
model $\mathscr{M}$, if and only if $G_{C}$ is a Directed Acyclic
Graph (DAG).
\end{thm}
\begin{proof}
If $G_{C}(V_{C},E_{C})$ is a DAG, then there exists a topological ordering
of its vertices $C_{1}, C_{2},\ldots, C_{k}$, that is, an ordering
such that for all $(C_{i},C_{j})\in E_{C}$, $i\leq j$ holds. In
view of this fact, since $C_{1}$ is the first in this ordering, there
will be only outgoing edges from the vertices of $C_{1}$, to those
of $V_{C}\backslash\{C_{1}\}$. Hence, according to Lemma \ref{lem:2-PFDCondition}
removal of the edges originating from $C_{1}$ and terminating to
$V_{C}\backslash\{C_{1}\}$ is a 2\textendash PFD (see Figure \ref{fig:C1PFD}).
\begin{figure}[h]
\begin{centering}
\begin{center}
\begin{tikzpicture}
\node at (0.5,1) [] {$C_1$};
\draw (0, 0) -- (0, 2) -- (1, 2) --
        (1, 1.75) -- (0.5, 1.75) -- (0.5, 1.5) --
        (2, 1.5) -- (2, 1) -- (1, 1) --
        (1, 0.5) -- (1.5, 0.5) -- (1.5, 0) -- (0, 0) ;

\node at (5,1) [] {$V_C\backslash \{C_1\}$};
\draw (3.5, 0) -- (6, 0) -- (6, 2) -- (3, 2) --
        (3, 1.75) -- (2.5, 1.75) -- (2.5, 1.5) --
        (4, 1.5) -- (4, 1) -- (3, 1) --
        (3, 0.5) -- (3.5, 0.5) -- (3.5, 0);

\draw[dashed] (4, 1.5) -- (6, 1.5);
\draw[dashed] (3.5, 0.5) -- (6, 0.5);

\draw[->,dashed] (1.25, 1.75) -- (2.4, 1.75);
\draw[->,dashed] (1.5, 0.75) -- (2.9, 0.75);
\draw[->,dashed] (1.75, 0.25) -- (3.15, 0.25);
\end{tikzpicture}
\par\end{center}
\par\end{centering}
\caption{2\textendash PFD of $C_{1}$, $\ensuremath{V_{C}\backslash\{C_{1}\}}$.}
\label{fig:C1PFD} 
\end{figure}
Now, if we denote by $G_{C}^{\prime}$ the subgraph of $G_{C}$ induced
by $V_{C}\backslash\{C_{1}\}$, we may note that $C_{2},\ldots,$$C_{k}$,
is a topological order of its vertices. Hence, $C_{2}$ can be detached
from $G_{C}^{\prime}$ through a 2\textendash PFD following a similar
procedure as above. Thus, after $k-1$ recursive applications of 2\textendash PFD's,
utilizing appropriate subsets of $\bar{E_{i}}$, we obtain a decomposition
of the model $\mathscr{M}$ into $k$ weakly connected components
$C_{1},C_{2},\ldots,C_{k}$, which is a $k$\textendash PFD.

Conversely, assume that the removal of the set of edges $\bar{E_{i}}$,implies
a $k$\textendash PFD of the model and let $G_{C}(V_{C},E_{C})$ be
the corresponding CCG. As explained in Remark \ref{rmk:k-PFD_as_Tree}
a $k$\textendash PFD of a model can be represented by a full, ordered,
binary tree $T$. Moreover, in view of the way that the left and right
children are assigned in each node of $T$, it is easy to verify that
if $(C_{i},C_{j})\in E_{C}$ then $C_{i}$ will appear on $T$, to
the left of $C_{j}$. Hence, if we order the leafs of $T$ starting
from the leftmost one moving to the right, we get a total order $C_{1},C_{2},\ldots,C_{k}$,
which is clearly a topological ordering of $G_{C}(V_{C},E_{C})$.
Thus, $G_{C}(V_{C},E_{C})$ is acyclic.
\end{proof}

\section{\label{sec:MaximalPFD}Maximal PFD along a spatial direction}

In the previous section, a characterization of physically feasible
decompositions along a particular spatial direction was given in
terms of the absence of cycles on the implied CCG. In this section,
we propose a method to derive such a maximal acyclic CCG, as the condensation
of the graph resulting after making edges not aligned to the chosen
direction, bidirectional. In this respect we introduce the following definitions.
\begin{defn}[Maximal PFD]
Let $\mathscr{M}$ be a toy model and let $G(V,E)$ be the associated
directed graph.  The removal of a set of edges $\bar{E}_{i}\subseteq E_{i}$,
along a spatial direction $\hat{d}_{i}$, implies a\emph{ maximal
PFD of the model along $\hat{d}_{i}$}, if the implied CCG is maximal,
that is, any set of edges $\bar{E}_{i}^{\prime}$, such that $\bar{E}_{i}\subseteq\bar{E}_{i}^{\prime}\subseteq E_{i}$,
implies the same CCG with $\bar{E}_{i}$.
\end{defn}
\begin{defn}[Projection along a direction]
\label{def:ProjectionAlongDir}Let $\mathscr{M}$ be a toy model,
$G(V,E)$ its associated directed graph and let $E_{i}\subseteq E$
be the set of all edges along the spatial direction $\hat{d}_{i}$.
We define the \emph{projection of $G(V,E)$ along the direction $\hat{d}_{i}$},
to be the graph $G_{i}(V,E\cup R_{i})$, where $R_{i}$ contains all
the edges of $G$ not in $E_{i}$, reversed, that is $R_{i}=\{(u,v):(v,u)\in E\backslash E_{i}\}$. 
\end{defn}
We illustrate the above notion via the following example. 
\begin{example}
\label{exa:ProjectionsAlongDirs}Consider the model presented in Example
\ref{exa:2PFDvsNoPFD} and the corresponding graph $G(V,E)$ shown
in Figure \ref{fig:Model1_Graph}. According to Definition \ref{def:ProjectionAlongDir}, the projection
of $G(V,E)$ along $\hat{d}_{1}$ and $\hat{d}_{2}$ are shown in
Figure \ref{fig:Model1_Projections}.
\end{example}
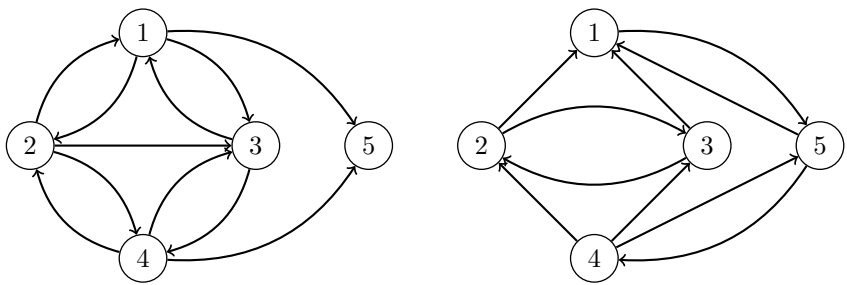
\begin{figure}[h]
\begin{centering}
\begin{center}
\begin{tikzpicture}
  \SetGraphUnit{1.5}
  \GraphInit[vstyle=Normal]

\begin{scope}
 \Vertex[L=$1$]{V1}
  \SOWE[L=$2$](V1){V2}
  \SOEA[L=$4$](V2){V4}
  \SOEA[L=$3$](V1){V3}
  \EA[L=$5$](V3){V5}

  \Edge[style={->}](V2)(V3)
  \Edge[style={bend left, ->}](V1)(V5)
  \Edge[style={bend right, ->}](V4)(V5)
  \Edge[style={bend left, ->}](V2)(V1)
  \Edge[style={bend left, ->}](V4)(V2)
  \Edge[style={bend left, ->}](V3)(V1)
  \Edge[style={bend left, ->}](V4)(V3)
  \Edge[style={bend left, ->}](V1)(V2)
  \Edge[style={bend left, ->}](V2)(V4)
  \Edge[style={bend left, ->}](V1)(V3)
  \Edge[style={bend left, ->}](V3)(V4)

\end{scope}

\begin{scope}[xshift=6 cm]
  \Vertex[L=$1$]{V1}
  \SOWE[L=$2$](V1){V2}
  \SOEA[L=$4$](V2){V4}
  \SOEA[L=$3$](V1){V3}
  \EA[L=$5$](V3){V5}

  \Edge[style={bend left, ->}](V2)(V3)
  \Edge[style={bend left, ->}](V1)(V5)
  \Edge[style=->](V4)(V5)
  \Edge[style={bend left, ->}](V3)(V2)
  \Edge[style=->](V5)(V1)
  \Edge[style={bend left, ->}](V5)(V4)

  \Edge[style=->](V2)(V1)
  \Edge[style=->](V4)(V2)
  \Edge[style=->](V3)(V1)
  \Edge[style=->](V4)(V3)

\end{scope}
\end{tikzpicture}

\par\end{center}
\par\end{centering}
\caption{The projections of $G(V,E)$ along $\hat{d}_{1}$ (left), $\hat{d}_{2}$
(right).}
\label{fig:Model1_Projections} 
\end{figure}

We proceed now to the main result of the present section. 
\begin{thm}
\label{thm:CondensationGivesMaxPFD}Let $\mathscr{M}$ be a toy model, let
$G(V,E)$ be its associated directed graph\emph{ and } let $G_{co}^{i}(V_{co}^{i},E_{co}^{i})$ be
the condensation of the projection $G_{i}(V,E\cup R_{i})$ of $G(V,E)$,
along \emph{$\hat{d}_{i}$}. Then, $G_{co}^{i}(V_{co}^{i},E_{co}^{i})$
is a CCG corresponding to a maximal PFD along $\hat{d}_{i}$.
\end{thm}
\begin{proof}
Define the set of edges whose endpoints lie on two distinct strongly
connected components of $G_{i}(V,E\cup R_{i})$, that is
\[
\bar{E}_{i}=\{(u,v)\in E:u\in C_{k},v\in C_{l}\textrm{ where }C_{k},C_{l}\in V_{co}^{i}\textrm{ and }k\neq l\}.
\]
Note that if either $(u,v)\in E \backslash E_{i}$ or $(u,v)\in R_{i}$, then
$u,v$ lie on the same strongly connected component of $G_{i}(V,E\cup R_{i})$,
because there are edges connecting them in both directions. Thus,
$\bar{E}_{i}\subseteq E_{i}$. 

If any two vertices $u,v\in V$ lie on the same strongly connected
component of $G_{i}(V,E\cup R_{i})$, then there exists a directed
path from $u$ to $v$, whose intermediate vertices lie on the same
strongly connected component with $u,v$. Every edge on the path which
is in $R_{i}$, can be replaced by its ``reverse'', which lies in
$E\backslash E_{i}\subseteq E\backslash\bar{E}_{i}$. The rest of
the edges on the path, not in $R_{i}$, obviously cannot be in $\bar{E}_{i}$,
since the latter contains edges whose endpoints lie on two distinct
strongly connected components of $G_{i}(V,E\cup R_{i})$. Hence, any
two vertices $u,v\in V$ lying on the same strongly connected component
of $G_{i}(V,E\cup R_{i})$, can be connected via an undirected path,
which lies entirely on the same weakly connected component as $u,v$,
using only edges from $E\backslash\bar{E}_{i}$. Thus, all vertices
lying on the same strongly connected component of $G_{i}(V,E\cup R_{i})$,
belong to the same weakly connected component of $G(V,E\backslash\bar{E}_{i})$. 

Conversely, if any two vertices $u,v\in V$ lie on the same weakly
connected component of $G(V,E\backslash\bar{E}_{i})$, then there
exists an undirected path from $u$ to $v$, whose intermediate vertices
are on the same weakly connected component with $u,v$. Our aim is
to show that there exists a directed path from $u$ to $v$ in $G_{i}(V,E\cup R_{i})$.
In this respect, the edges on the undirected path having the correct
orientation, that is from $u$ to $v$, can be used to form the directed
path. If an edge on the undirected path belongs to $E\backslash E_{i}$
and is oriented from $v$ to $u$, then it can be replaced in $G_{i}(V,E\cup R_{i})$
by its ``reverse'' which belongs to $R_{i}$. On the other hand,
if an edge on the undirected path belongs to $E_{i}$, then both its
endpoints must lie in the same strongly connected component of $G_{i}(V,E\cup R_{i})$,
otherwise this edge should be in $\bar{E}_{i}$, whose elements have
been removed from $G(V,E\backslash\bar{E}_{i})$. In view of this,
if such an edge does not have the desired orientation (i.e. from $u$
to $v$), we can find a directed path in $G_{i}(V,E\cup R_{i})$,
with the correct orientation, to replace it. Thus, any two vertices
lying on the same weakly connected component of $G(V,E\backslash\bar{E}_{i})$,
belong to the same strongly connected component of $G_{i}(V,E\cup R_{i})$.

In view of the above discussion, it is clear that the strongly connected
components of $G_{i}(V,E\cup R_{i})$ coincide with the weakly connected
components of $G(V,E\backslash\bar{E}_{i})$. Thus, $V_{co}^{i}$
is the vertex set of the CCG implied by the removal of the edges of
$\bar{E}_{i}$ from $G(V,E)$. Further, it is straightforward to verify
that the set of edges $E_{co}^{i}$ are exactly the edges of the CCG
implied by the removal of the edges of $\bar{E}_{i}$ from $G(V,E)$.
Thus, $G_{co}^{i}(V_{co}^{i},E_{co}^{i})$ is a CCG corresponding
to the removal of the edges of $\bar{E}_{i}$. Since the condensation
graph of any directed graph is a DAG, the removal of the edges of
$\bar{E}_{i}$, implies a $k$\textendash PFD of the model, where
$k=\left|V_{co}^{i}\right|$. 

To show that the removal of the edges of $\bar{E}_{i}$, implies a
maximal PFD along $\hat{d}_{i}$, assume there exists a set of edges
$\bar{E}_{i}^{\prime}$, such that $\bar{E}_{i}\subseteq\bar{E}_{i}^{\prime}\subseteq E_{i}$,
implying a PFD of the model along $\hat{d}_{i}$. Consider an edge
$(u,v)\in\bar{E}_{i}^{\prime}\backslash\bar{E}_{i}$, whose endpoints
lie on distinct weakly connected components $C_{u},C_{v}$, in $G(V,E\backslash\bar{E}_{i}^{\prime})$,
such that $u\in C_{u}$ and $v\in C_{v}$. Clearly, since $(u,v)\notin\bar{E}_{i}$,
it is present in $G(V,E\backslash\bar{E}_{i})$ and both $u,v$ lie
in the same weakly connected component of the latter. In this case,
it is evident from the discussion above that $u,v$ must lie on the
same strongly connected component of $G_{i}(V,E\cup R_{i})$. Thus,
there exists a directed path from $v$ to $u$ in $G_{i}(V,E\cup R_{i})$.
Now, since $u\in C_{u}$ and $v\in C_{v}$ in $G(V,E\backslash\bar{E}_{i}^{\prime})$,
there exists at least one edge $(v^{\prime},u^{\prime})\in\bar{E}_{i}^{\prime}$,
in the directed path from $v$ to $u$, such that $u^{\prime}\in C_{u}$
and $v^{\prime}\in C_{v}$, otherwise $C_{u},C_{v}$ would not be
distinct. Hence, the weakly connected components $C_{u},C_{v}$ are
connected in the CCG implied by the removal of the edges of $\bar{E}_{i}^{\prime}$,
via to opposite edges, which in turn implies that such a removal does
not imply a PFD. Having arrived at a contradiction, we conclude that
there exists no edge in $\bar{E}_{i}^{\prime}\backslash\bar{E}_{i}$,
thus $\bar{E}_{i}^{\prime}=\bar{E}_{i}$.
\end{proof}
Theorem \ref{thm:CondensationGivesMaxPFD} essentially provides 
a method to obtain a maximal PFD of a given model along a spatial 
direction $\hat{d}_{i}$. According to the above result the CCG 
corresponding to a maximal PFD along $\hat{d}_{i}$ coincides with
the condensation, $G_{i}(V,E\cup R_{i})$, of $G(V,E)$ along this particular
direction. Thus, the components into which a maximal PFD decomposes the model, 
coincide with the strongly connected components of the corresponding projection.
The computation of the strongly connected components can be accomplished in 
linear time, using Kosaraju's algorithm \cite{CormenEtAl2001,Sharir1981}, 
Tarjan's strongly connected components algorithm \cite{Tarjan1972} or Dijkstra's 
path based strong component algorithm \cite{Dijkstra1976}. Moreover, Kosaraju's 
and Tarjan's algorithms also compute a reverse topological ordering of the 
strongly connected components of the graph on which it is applied. The topological
ordering computed by these algorithms dictates the order under which the components
detected can be detached from the model in the process of a step-by-step decomposition
along the chosen spatial direction. 
\begin{example}\label{exa:MaxPFDs}
Applying some strongly connected component computation algorithm on
the projections of $G(V,E)$ along $\hat{d}_{1}$, $\hat{d}_{2}$,
given in Example \ref{exa:ProjectionsAlongDirs}, we get respectively
the condensations shown in Figure \ref{fig:Model1_Condenstations}.
\begin{figure}[h]
\begin{centering}
\begin{center}
\begin{tikzpicture}
  \SetGraphUnit{2}
  \GraphInit[vstyle=Normal]

\begin{scope}

  \Vertex[L=$C_1$]{C1}
  \EA[L=$C_2$](C1){C2}

  \Edge[style=->](C1)(C2)

\end{scope}

\begin{scope}[xshift=6 cm]
   \Vertex[L=$C_1^\prime$]{C1}

\end{scope}
\end{tikzpicture}

\par\end{center}
\par\end{centering}
\caption{The condensations of $G(V,E)$ along $\hat{d}_{1}$ (left), $\hat{d}_{2}$
(right).}
\label{fig:Model1_Condenstations} 
\end{figure}
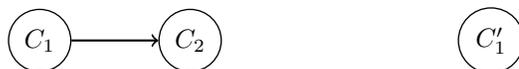
Clearly, the condensed graph corresponding to the projection along
$\hat{d}_{1}$, coincides with the CCG shown in Figure \ref{fig:Model1_CCG_2PFD}
and clearly implies a 2\textendash PFD of the model along this direction.
On the other hand the condensed graph corresponding to the projection
along $\hat{d}_{2}$, consists of only one component, indicating that
a  $k$\textendash PFD, for $k\geq2$, along $\hat{d}_{2}$ is not possible. 
\end{example}

\section{\label{sec:HierarchicalPFD}Hierarchical PFD of toy models and assembly instructions generation}
In the present section an outline of the procedure to obtain a recursive, physically feasible decomposition of a given toy model $\mathscr{M}$, is proposed. The key step of the procedure presented in what follows, is based on both the theoretical analysis presented in Section \ref{sec:PFD}, and the use of well established algorithmic tools for the detection of strongly connected components in directed graphs, as shown in Section \ref{sec:MaximalPFD}. While each step of the procedure results in a flat collection of weakly connected components, corresponding to a maximal PFD along some given spatial direction, the outcome of the overall procedure will be a hierarchical model of components, i.e. a rooted tree, having as its top level component the toy model $\mathscr{M}$ itself, and bottom level elements each of the constituent blocks of the model. Having obtained a hierarchical decomposition of the model, some appropriate tree traversal algorithm may be applied to reverse the decomposition process and produce a step-by-step assembly manual. This procedure is outlined at the end of this section.

Using the setup of the previous sections, assume that $G(V,E)$ is the directed graph associated to the model $\mathscr{M}$. Assume also that each directed edge in $E$ is aligned to one of the $p$ distinct spatial directions $\hat{d_{i}}$, $i=1,2,\ldots,p$. Finally, assume that $\textrm{MaxPFD}(C,i)$ is a readily made function taking as its first argument a weakly connected component of $G(V,E)$ and as its second argument an integer $i=1,2,\ldots,p$. The function returns an ordered list of components $C_1, C_2,\ldots, C_k$, $k\geq 1$, into which $C$ can be decomposed as the result of a Maximal PFD along the direction $\hat{d_{i}}$. According to the results of Section \ref{sec:MaximalPFD} such a function can be implemented using well known, linear-time, strongly connected components detection algorithms. 

With this background we define the function $\textrm{HMaxPFD}(C)$ which accepts as argument a weakly connected component of $G(V,E)$, $C$, and returns a hierarchical decomposition of the model $\mathscr{M}$. The function HMaxPFD is outlined as follows:

\subsubsection*{HMaxPFD(C)}
\begin{itemize}
	\item Call $\textrm{MaxPFD}(C,i)$ for $i=1,2,\ldots,p$. 
	\item If for at least one  $i=1,2,\ldots,p$, the number of components $C_1, C_2, \ldots, C_k$, returned by the respective MaxPFD, is greater than 1, then 
	\begin{itemize}
		\item For $j = 1 \ldots k$
		\begin{itemize}
			\item Call $\textrm{AppendChild}(C, C_j)$
			\item Call $\textrm{HMaxPFD}(C_j)$
		\end{itemize}
	\end{itemize}
\end{itemize}

In the above pseudocode the function $\textrm{AppendChild}(C, C_j)$ is called, which is assumed to append the subcomponent $C_j$ to $C$, as its child in the hierarchy of the intended decomposition. To implement this in practice, would require each of the discovered components to be able to maintain a list of pointers, pointing from each parent to its children components. The technical details of such an implementation are out of the scope of the present paper. Finally, when the argument of $\textrm{HMaxPFD}$ is a single vertex $v$ (which will necessarily be without edges), we define $\textrm{HMaxPFD}(v) = v$ and the hierarchical operations terminate there, to then pass to the next branch (if any). 

To obtain the tree corresponding to the hierarchical PFD of $\mathscr{M}$, with the associated graph $G(V,E)$, one has to invoke the function HMaxPFD, using the entire graph $G$ as its sole argument. To provide a worst case analysis of the complexity of the HMaxPFD algorithm, we first take into account that each run of $\textrm{MaxPFD}(C,i)$ is essentially a call of Tarjan's or a similar algorithm, whose time complexity is $O(|V|+|E|)$, where  $|V|,|E|$ are the number of nodes and edges of the graph to which it is applied. Considering the worst case scenario, the MaxPFD will be called at most $p$ times, until an actual decomposition, into two or more subcomponents is obtained. Moreover, at each level of the resulting hierarchical PFD tree, the total number of nodes (blocks) distributed along the components $C_1, C_2, \ldots, C_k$, will be at most $n$, where $n$ is total number of vertices in $G$ (blocks in $\mathscr{M}$). Thus, if we denote by $m$ the total number of edges in $G$, then the invocation of $\textrm{MaxPFD}(C_j,i)$, for $i=1,2,\ldots,p$, $j=1,2,\ldots,k$ will take at most $O(p(n+m))$ steps. Since the time complexity at each level of the tree is $O(p(n+m))$, the overall worst case complexity will occur on a PFD tree that has the maximum possible height, amongst all the PFD trees with $n$ leaves in total and whose non-leaf nodes have at least two children. This becomes evident if we take into account the fact that the leaves of a PFD tree are exactly the components of $G$ that can be no further decomposed, i.e. its individual blocks. In view of this, the maximum height PFD tree, will be a binary tree where every non-leaf node has exactly two children, out of which at least one is a leaf. The height of such a binary tree with $n$ leaves can be easily seen to be $n-1$. Thus, the worst case time complexity is $O(np(n+m))$.

We illustrate the above procedure in the following example.

\begin{example}\label{exa:HPFD}
	Consider the toy model of Example \ref{exa:2PFDvsNoPFD} and its associated graph shown in Figure \ref{fig:Model1_Graph}. Invoking 
HMaxPFD(G), the procedure will execute as follows:

\begin{itemize}
	\item Calling MaxPFD$(G,1)$ returns two components $C_{1},C_{2}$ where
	$C_{1},C_{2}$ consist of the vertices $\{1,2,3,4\}$ and $\{5\}$
	respectively.
	\item Since MaxPFD returned more than one component for $i=1$,
	\begin{itemize}
		\item For $j=1$,
		\begin{itemize}
			\item $C_{1}$ is appended as a child of $G$.
			\item HMaxPFD$(C_{1})$ is called. 
			\begin{itemize}
				\item MaxPFD$(C_{1},2)$ ($\hat{d}_{2}$ is the only direction available)
				returns three components, $C_{11},C_{12}$ and $C_{13}$, having as vertex sets $\{1\}$, $\{2,3\}$ and $\{4\}$ respectively. 
				\item Since MaxPFD returned more than one component for $i=2$, 
				\begin{itemize}
					\item For $j'=1$, 
					\begin{itemize}
						\item $C_{11}$ is appended as a child of $C_{1}$. 
						\item HMaxPFD$(C_{11})$ is called, returning $C_{11}$ since this is a single vertex. Recursion terminates.  
					\end{itemize}
					\item For $j'=2$, 
					\begin{itemize}
						\item $C_{12}$ is appended as a child of $C_{1}$. 
						\item HMaxPFD$(C_{12})$ is called.
			\begin{itemize}
				\item MaxPFD$(C_{12},1)$ ($\hat{d}_{1}$ is the only direction available here)
				returns two components, $C_{121}$ and $C_{122}$, having as vertex sets $\{2\}$ and $\{3\}$ respectively. 
				\item Since MaxPFD returned more than one component for $i=1$,  
				\begin{itemize}
					\item For $j''=1$, 
					\begin{itemize}
						\item $C_{121}$ is appended as a child of $C_{12}$. 
						\item HMaxPFD$(C_{121})$ is called, returning $C_{121}$ since this is a single vertex. Recursion terminates.  
					\end{itemize}
					\item For $j''=2$, 
					\begin{itemize}
						\item $C_{122}$ is appended as a child of $C_{12}$. 
						\item HMaxPFD$(C_{122})$ is called, returning $C_{122}$ since this is a single vertex. Recursion terminates.  
					\end{itemize}
				\end{itemize}
			\end{itemize}						
					\end{itemize}
					
					\item For $j'=3$, 
					\begin{itemize}
						\item $C_{13}$ is appended as a child of $C_{1}$. 
						\item HMaxPFD$(C_{13})$ is called, returning $C_{13}$ since this is a single vertex. Recursion terminates.  
					\end{itemize}
				\end{itemize}
			\end{itemize}
		\end{itemize}
		\item For $j=2$,
		\begin{itemize}
			\item $C_{2}$ is appended as a child of $G$.
			\item HMaxPFD$(C_{2})$ is called, returning $C_{2}$ since this is a single vertex. Recursion terminates. 
		\end{itemize}
	\end{itemize}
\end{itemize}

The resulting hierarchical PFD of the model is depicted in figure \ref{fig:HPFD}. The assembly instructions for the model can be recovered by applying a depth - first traversal, starting from the root node of the tree. 

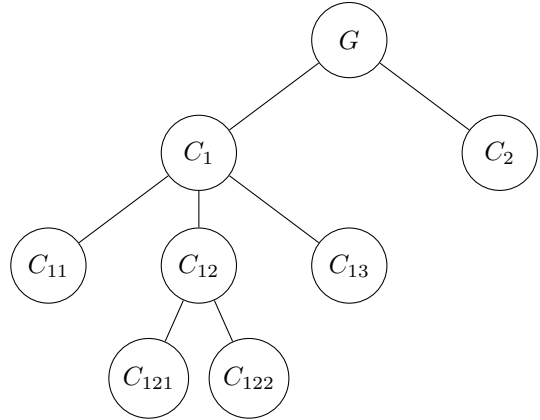
\begin{figure}[h]
	\begin{center}

			\begin{tikzpicture}[level/.style={sibling distance=40mm/#1},
								ball/.style = {circle,draw,minimum width=1cm}]
				\node [ball] (z){$G$}
					child {node [ball] (a) {$C_1$}
						child {node [ball] (b) {$C_{11}$}}
						child {node [ball] (c) {$C_{12}$}
							child {node [ball] (i) {$C_{121}$}}
							child {node [ball] (ii) {$C_{122}$}}}
						child {node [ball] (d) {$C_{13}$}}
					}
					child {node [ball] (a) {$C_2$}};
			\end{tikzpicture}
	\end{center}
	\caption{The hierarchical PFD of the toy model in Example \ref{exa:2PFDvsNoPFD}.}
	\label{fig:HPFD} 
\end{figure}
\end{example}

Having obtained a hierarchical decomposition of a toy model $\mathscr{M}$, which is essentially a tree structure like the one shown in figure \ref{fig:HPFD}, we can proceed to the composition of its nodes to reverse the PFD and produce the assembly instructions. This goal can be accomplished by employing a tree traversal algorithm, which respects the parent - child hierarchy, in the sense that each node is visited after its children. The necessity of the requirement regarding the priority of visits between parents and their children, emerges from the fact in order to assemble a component, from its constituent subcomponents, i.e. the children of the node in the tree, one has to assemble each child component first. 

An algorithm appropriate for this task could be a postorder traversal applied on the hierarchical PFD tree of $\mathscr{M}$. Preorder, inorder and postorder are well known traversal procedures that can be applied on ordered binary trees, i.e. rooted trees whose nodes have at most two children labeled as ``left" and ``right". A preorder traversal visits first the parent node, then traverses the left subtree and finally the right subtree. Respectively, the inorder traversal first traverses the left subtree of a node, then visits the node itself and finally traverses the right subtree. Postorder traversal, traverses first the left subtree of a node, then its right subtree and finally it visits the node itself. 

While inorder traversal may be ambiguous when applied to a general (non - binary) ordered tree, preorder and postorder traversals are well defined. Here we focus on the generalized version of the postorder traversal algorithm, which is applicable to non-binary trees. Given a node $p$ in such a tree, the postorder traversal procedure can be defined recursively as follows: 
\begin{itemize}
	\item Traverse the leftmost child of $p$, 
	\item Visit the node $p$, 
	\item Traverse the right sibling of $p$. 
\end{itemize}

The output of such an algorithm is a series of nodes ordered in such a way that parent nodes appear in the sequence after all their children. In view of this fact, the sequence generated by the postorder traversal can be used to generate the assembly instructions of the toy model $\mathscr{M}$. It should be noted that the leaf nodes of the hierarchical PFD tree represent individual toy blocks that require no assembly, thus they can be safely neglected in the instructions generation procedure. On the other hand, the internal nodes of the tree represent components of the toy model consisting of an assembly of individual blocks or other subcomponents, and thus they are the ones for which assembly instructions are needed. The ingredients required for the assembly of those components are no other than their child nodes in the hierarchy. If a subcomponent is used as a building block for a higher level component in the PFD hierarchy, then the former will precede the latter in the ordered sequence produced by the postorder traversal. Moreover, as noted in section \ref{sec:intro}, the exact geometry of the interconnections between any two blocks has been recorded beforehand. Thus, identifying the blocks from which a component is comprised, provides enough information to recover the exact geometric structure of each component. As a result, following the order dictated by the traversal, the assembly instructions of every subcomponent comprising a higher level component will appear earlier in the assembly procedure manual, leaving no room for inconsistencies in the flow of instructions.  

The procedure is illustrated in the following example.

\begin{example}\label{exa:HPFD-Assembly}
	Given the hierarchical PFD of Example \ref{exa:HPFD}, we may apply the postorder traversal procedure on the tree shown in figure \ref{fig:HPFD}. The outcome of the algorithm is the following sequence of nodes:
	\[
	C_{11}, C_{121}, C_{122}, C_{12}, C_{13}, C_{1}, C_{2}, G
	\]
	As explained above the leaf nodes $C_{11}$, $C_{121}$, $C_{122}$, $C_{13}$, $C_{2}$ can be safely omitted, since they require no assembly. Doing so, the sequence of the remaining nodes consists only of $C_{12}$, $C_{1}$ and $G$, in that particular order. 
	
	Thus, the assembly instructions manual in this case should consist of the following three steps: 
	\begin{enumerate}
		\item Show how $C_{12}$ is assembled from its children nodes $C_{121}, C_{122}$.
		\item Respectively, show how $C_{1}$ can be assembled from $C_{11}, C_{12}, C_{13}$.
		\item Finally, show how $G$ is assembled from using $C_{1}, C_{2}$.
	\end{enumerate}

\end{example}	

\section{\label{sec:Conclusions}Conclusions}

In this note we study the problem of automatically producing step-by-step
assembly instructions for Engino$^{\circledR}$ toy models. The assembly
manual of a toy model can be generated by reversing the decomposition
process of the model to its constituent blocks. As explained in Section
\ref{sec:Prerequisites} the disassembly process may under certain circumstances be blocked
due to the presence of particular geometric structures in the interconnections
between blocks. To avoid such situations we propose a graph theoretic
framework for the analysis of the problem and provide a characterization
of the decompositions that are physically feasible. Moreover, a procedure
to obtain maximal physically feasible decompositions along a given
geometric direction is presented, which can be implemented using well
known, linear time, algorithms for the detection of strongly connected
components in directed graphs. Based on these results, an algorithmic 
procedure for the hierarchical decomposition of a given toy model, which takes
into account the physical feasibility of the intermediate steps, is proposed. 
The final goal of generating a sequence of assembly instructions for the model
is accomplished, by applying a postorder traversal of the hierarchical decomposition 
tree, from which a step-by-step series of instructions can be easily recovered.

As for future extensions that could stem from our presented approach, and future challenges to be tackled, we remark the following:

\begin{itemize}
\item First of all, notice that the connection principle for ENGINO blocks is mainly of binary type (just like those of LEGO and other toy systems), in the sense that even though some types of blocks have several male and/or female connectors, thus allowing for several ways of connecting two given blocks, any connection between two blocks is achieved by matching at least one pair of male-female connections, resulting in a finite set of possible relative spatial configurations between the blocks. An important exception to this is the freely-rotating pivoting connection, which allows for a continuous choice of the pivot angle, so the set of relative spatial configurations becomes infinite. In this paper we have focused on the binary type of connections because of the resulting finiteness of the set of possible spatial configurations, which allows us to tackle the problem by defining the connection directions $\hat{d}_i, \quad i = 1, \ldots, p$. In future work, the feature of pivoting connections will be added to our Physically Feasible Decomposition, based on the fact that pivoting connections during the assembly process must be geometrically feasible, in the sense that small displacements associated with the rotation degree of freedom about the pivot point must be allowed to happen. The main difficulty lies in extending the current definition of the ``fixed'' connection directions $\hat{d}_i, \quad i = 1, \ldots, p$, which will have to depend on the continuous ``pivoting'' degrees of freedom.

\item From the previous point it follows that our method can be applied to several toy systems, and even beyond that to industrial assembly processes \cite{Na88,Wa09} with binary-type connections as defined above. The main advantage of our method is that it requires very little geometrical and physical information about the connecting pieces. This is, at the same time, the main limitation of the method. For example, it does not apply to assemblies that require three or more hands \cite{Sn94}, and more generally it does not deal with cases when force and torque balances are relevant, as in the problems of grasping parts (form closure, force closure, etc.) \cite{Wi94}. However, this does not mean our method cannot be used in combination with these and other advanced assembly features. In fact, our method could be included as a complementary module in (dis-)assembly process planning for existing products in industry. For example, the feature of linearizability \cite{Wi94}, of practical importance in assembly lines, could be incorporated into our method because it is related to the distribution of internal nodes and leaves in our hierarchical PFD graphs. And, with a little bit of imagination, our method could potentially find its utility as a module in the recently discovered molecular assembly processes \cite{Ka17}, because these processes are characterized by constrained geometric arrangements, local interactions and reduced reactivity.
\end{itemize}

\begin{acknowledgements}
This work originated from our participation in the 125th European Study Group with Industry (1st Study Group with Industry in Cyprus). We thank the Mathematics for Industry Network (MI-NET, www.mi-network.org), COST Action TD1409 for generous funding and support with the logistics of this first Study Group with Industry in Cyprus. MDB acknowledges support from Science Foundation Ireland under research grant number 12/IP/1491.

We would also like to thank Costas Sisamos, founder and CEO of Engino Ltd, for the detailed exposition of the problem, his valuable insight on it and his comments on the present paper. Finally, we would like to thank the editor and the anonymous referees for taking time in reading and suggesting modifications to the paper. We highly appreciate it, as the comments have been very useful in improving the paper.
\end{acknowledgements}


\begin{thebibliography}{1}

\bibitem{Agrawala2003} \textsc{Agrawala, M. \etal} (2003) {Designing Effective Step-by-step 
Assembly Instructions}.  \textit{ACM Transactions on Graphics (TOG) - Proceedings of 
ACM SIGGRAPH 2003} \textbf{22}, 828--37.

\bibitem{BangGutin2008} \textsc{Bang-Jensen, J. \& Gutin, G.~Z.}  (2008) \textit{Digraphs: 
Theory, Algorithms and Applications}. Springer Publishing Company Inc. 

\bibitem{BondyAndMurty1976} \textsc{Bondy, J.~A. \& Murty, U.~S.~R.} (1976) \textit{Graph Theory
with Applications}. North Holland.

\bibitem{CormenEtAl2001} \textsc{Cormen, T.~H., Leiserson, C.~E., Rivest, R.~L. 
\& Stein, C.} (2001) \textit{Introduction to Algorithms}. Second Edition. MIT
Press and McGraw-Hill.

\bibitem{Dijkstra1976} \textsc{Dijkstra, E.} (1976) \textit{A Discipline of Programming}. NJ: Prentice Hall.

\bibitem{Hsu2011} \textsc{Hsu, Y.-Y., Tai, P.-H., Wang, M.-W. \& Chen, W.-C.} (2011) 
{A knowledge-based engineering system for assembly sequence planning}. \textit{Int J Adv Manuf Technol} \textbf{55}, 763--82.

\bibitem{Ka17} \textsc{Kassem, S., Lee, A. T., Leigh, D. A., Marcos, V., Palmer, L. I. \& Pisano, S.} (2017) {Stereodivergent synthesis with a programmable molecular machine}. \textit{Nature} \textbf{549}, 374--378.

\bibitem{Lambert2003} \textsc{Lambert, A.~J.~D.} (2003) {Disassembly sequencing: A survey}. \textit{International Journal of Production Research} \textbf{41}, 3721--59.

\bibitem{Li2008} \textsc{Li, W.,  Agrawala, M., Curless, B. \& Salesin, D.} (2008){Automated generation of interactive 3D exploded view diagrams}. \textit{ACM Transactions on Graphics (TOG) - Proceedings of ACM SIGGRAPH 2008} \textbf{27}, 101.

\bibitem{Na88}
\textsc{Natarajan, B.~K.} (1988) {On planning assemblies}. In \textit{Proceedings of the fourth annual symposium on Computational geometry}, ACM, 299--308.

  \bibitem{Peysakhov}
  \textsc{Peysakhov, M., Galinskaya, V. \& Regli, W.~C.} (2000) Representation and evolution of lego-based assemblies. \textit{In AAAI/IAAI} (p. 1089).


\bibitem{Sharir1981} \textsc{Sharir, M.} (1981) A strong connectivity algorithm
and its applications to data flow analysis. \textit{Computers and Mathematics
with Applications} \textbf{7}, 67--72.

\bibitem{Sn94} \textsc{Snoeyink, J. \& Stolfi, J.} (1994) Objects that cannot be taken apart with two hands. \textit{Discrete \& Computational Geometry} \textbf{12}, 367--384.

\bibitem{Tarjan1972} \textsc{Tarjan, R.~E.} (1972) Depth-first search and linear
graph algorithms.  \textit{SIAM Journal on Computing} \textbf{1}, 146--60.
\bibitem{Wa09} \textsc{Wang, L., Keshavarzmanesh, S., Feng, H.~Y. \& Buchal, R.~O.} (2009) Assembly process planning and its future in collaborative manufacturing: a review. \textit{The International Journal of Advanced Manufacturing Technology} \textbf{41}, 132--144.

\bibitem{Wi94} \textsc{Wilson, R.~H. \& Latombe, J.~C.} (1994) Geometric reasoning about mechanical assembly. \textit{Artificial Intelligence} \textbf{71}, 371--396.

\end{thebibliography}
\end{document}